\documentclass[10pt]{article}

\usepackage{enumerate}
\usepackage{amsmath}
\usepackage{amsfonts}
\usepackage{amssymb}
\usepackage{amsthm}
\title{ \bf Abstract Storage Devices }\author{Robert
K{\"o}nig\thanks{Centre for Quantum Computation, University of
Cambridge, United Kingdom,  E-mail: r.t.koenig@damtp.cam.ac.uk} \and Ueli
Maurer\thanks{Department of Computer Science, ETH Zurich, 8092
Zurich, Switzerland, E-mail: maurer@inf.ethz.ch} \and Stefano
Tessaro\thanks{Department of Computer Science, ETH Zurich, 8092
Zurich, Switzerland, E-mail: tessaros@inf.ethz.ch}}
\date{}

\newcommand{\D}{\mathcal{D}}
\newcommand{\notequiv}{\,/\kern-.6em\hbox{$\equiv$}\,}

\newcommand{\state}[0]{s}

\newcommand{\Stateof}[1]{\mathcal{S}^{#1}}

\newcommand{\func}[0]{g}

\newcommand{\set}[1]{\mathcal{#1}}

\newcommand{\BIGOP}[1]{\mathop{\mathchoice%
{\raise-0.22em\hbox{\huge $#1$}}%
{\raise-0.05em\hbox{\Large $#1$}}{\hbox{\large $#1$}}{#1}}}
\newcommand{\BIGOPP}[1]{\mathop{\mathchoice%
{\raise-0.22em\hbox{\huge $#1$}}%
{\raise-0.05em\hbox{\Large $#1$}}{\hbox{\large $#1$}}{#1}}}

\newcommand{\bigtimes}{\BIGOP{\times}}

\newcommand{\range}[1]{\textup{range}(#1)}
\renewcommand{\ker}[1]{\textup{ker}(#1)}

\newcommand{\refines}[0]{\sqsubseteq}
\newcommand{\srefines}[0]{\sqsubset}

\newcommand{\contained}[0]{\preceq}

\newcommand{\Part}[1]{\Pi\left(#1\right)}
\newcommand{\Partof}[1]{\Pi^{#1}}

\newcommand{\id}{\textit{id}}

\newcommand{\G}{\mathcal{G}}
\newcommand{\V}{\mathcal{V}}
\newcommand{\E}{\mathcal{E}}

\newtheorem{theorem}{Theorem}

\newtheorem{lemma}[theorem]{Lemma}
\newtheorem{proposition}[theorem]{Proposition}
\newtheorem{corollary}[theorem]{Corollary}

\theoremstyle{remark}
\newtheorem{claim}{Claim}

\theoremstyle{definition} 
\newtheorem{definition}{Definition}

\begin{document}
\maketitle

\begin{abstract}
A quantum storage device differs radically from a conventional
physical storage device. Its state can be set to any value in a
certain (infinite) state space, but in general every possible read
operation yields only partial information about the stored state.

The purpose of this paper is to initiate the study of a combinatorial
abstraction, called \emph{abstract storage device} (ASD), which models
deterministic storage devices with the property that only partial
information about the state can be read, but that there is a degree of
freedom as to which partial information should be retrieved. 

This concept leads to a number of interesting problems which we address,
like the reduction of one device to another device, the equivalence of
devices, direct products of devices, as well as the factorization of a
device into primitive devices. We prove that every ASD has an equivalent
ASD with minimal number of states and of possible read operations. Also, we
prove that the reducibility problem for ASD's is $\mathcal{NP}$-complete,
that the equivalence problem is at least as hard as the graph isomorphism
problem, and that the factorization into binary-output devices (if it
exists) is unique.

\vspace*{0.5em}
\end{abstract}

\section{Introduction}
\label{sec:introduction}

\subsection{Motivation}
\label{sec:motivation}

The term storage device is conventionally used for a physical device
with a \emph{write} and a \emph{read} operation which can store data
reliably, i.e., with the property that the read operation yields an
exact copy of the data previously written into the device. In this
paper, we consider a generalized type of storage devices for which
the write operation consists of setting the device's state to some
value in the state space, and the subsequent read operation consists
of performing some measurement and provides some (usually only
partial) information about the state.

Such a storage device is a relevant special case of a general physical
system. The state of such a system can in general not be measured
exactly. This may be due to intrinsic reasons. For example, it is
inherently impossible to perfectly measure a quantum
state\footnote{unless it is known to be one of a set of orthogonal
states}.  Also, practical constraints (like the required efficiency)
may impose an unavoidable inaccuracy to the measurement of the
state. For instance, a tape only allows to efficiently retrieve its
content \emph{locally} by sequentially accessing the small portion of
it being of interest.

The task of a conventional storage device (e.g., a hard disc) is to
store information reliably. The design goal of such a system is
therefore to define a finite subset of its state space (as large as
possible) such that the available read operation allows to distinguish
different such states with negligible error probability.  For this
reason, a conventional storage device is characterized by its
\emph{storage capacity}, i.e., the number of bits that can be stored
reliably in it.

Here, we take a more general approach to storage devices, by modeling
explicitly the fact that, on one hand, a read operation provides only
partial information about the state, but that, on the other hand, many
different such read operations can be available.  We typically assume that
only one of these operations can be performed, but that the choice is free.

There are different motivations for considering such a setting.  A
first motivation is \emph{quantum cryptography} or, more precisely,
\emph{privacy amplification}, the last step of a quantum key agreement
protocol (see~\cite{KoMaRe03}). In simplified terms, an adversary is
assumed to have access to a bit string~$S$ of length~$n$, shared by
the legitimate users, and can store information about $S$ in a
$2^k$-dimensional quantum device, where~$k<n$. Since the (reliable)
storage capacity of the device is only~$k$, the adversary cannot store
$S$ perfectly. Later, the legitimate users select a hash function~$h$
from $n$ bits to $t$ bits (where $t<k$) at random from a class of such
functions, and the adversary can now perform a measurement of the
quantum state, {\em depending} on the choice of $h$. In this context,
the goal is to prove that every such measurement yields only a
negligible amount of information about $h(S)$. One can naturally
generalize the setting of privacy amplification to other types of
storage devices.

As an additional motivating example, one can consider the following game:
An entity, say Alice, is given access to an $n$-bit string $s =
[s_1,\ldots, s_n]$ about which she stores partial information. Later, she
will learn a function~$f$ drawn from a given set and will have to guess the
output~$f(s)$. For example, this set of functions might consist of all
linear predicates~$a_1 s_1+ \cdots + a_n s_n \pmod{2}$ for some
$a_1,\ldots, a_n \in \{0,1\}$. A natural question one may ask is finding
the minimal amount of reliable storage required to win this game. More
generally, one may be interested in deciding whether keeping information
about~$s$ in a certain storage device suffices to succeed in the game.
Also, one may even want to compare such games in the sense of determining
whether one game is strictly more difficult than another one. Similar
games, which may be of independent interest, occur in the security analyses
of certain cryptographic schemes.

The purpose of this paper is to initiate the study of a combinatorial
abstraction, called {\em abstract storage device (ASD)}, which models the
described property that only partial information about the state can be
read, but that there is a degree of freedom as to which partial information
should be retrieved.  Both generalized storage devices as well as the above
game can be described as an ASD.  Here we only consider {\em deterministic}
storage devices, i.e., we analyze the case with no error probability. This
is similar in spirit to the investigation of the {\em zero-error
  capacity}~\cite{Sha56} in communication theory. Like there, the
treatments of the zero-error and the negligible-error cases are quite
different and deserve separate investigation.

A natural problem related to the above game is \emph{reducibility} of
devices, which asks for deciding whether a certain device can be
implemented by a second one. Additionally, this concept directly
implies a notion of \emph{equivalence} for devices.

In many branches of science, a common approach to analyze complex
objects is to represent such objects as compositions of simpler and
better-understood ones. From a mathematical point of view, product
factorizations of discrete structures have been studied in many forms
in the past, for instance in the context of graph products and of
finite relational structures (see~\cite{PG00,Jo66} for respective
surveys). Along similar lines, one can introduce \emph{direct
products} of ASD's and study direct product factorizations into
simpler primitive devices.

\subsection{Contributions and Outline of This Paper} 
The main contribution of this paper is the introduction of~abstract
storage devices (ASD). Section~\ref{sec:asd_definition} presents this
abstraction and gives some examples. There, we also define direct
products of ASD's. Moreover, we state the problems of reducibility and
equivalence of ASD's in Section~\ref{sec:reducibility}.

We prove in Section~\ref{sec:minimality} that every ASD has an
equivalent ASD which has both a minimal number of states {\em and} a
minimal number of possible read operations, and we discuss properties
of such devices with respect to reducibility and equivalence.

Also, we present and analyze relevant quantities related to ASD's. The
\emph{storage capacity} provides a measure of the amount of
information that can be reliably stored in a device, while the
\emph{state complexity} characterizes the minimal amount of reliable
storage needed to simulate the device. Finally, the \emph{perfectness
index} of an ASD's is the minimal number of read operations needed to
entirely retrieve the state of a device. These quantities yield
easily-verifiable necessary conditions for reducibility, and
Section~\ref{sec:orderpreserving} is devoted to their discussion.

In Section~\ref{sec:complexity}, we prove the general problem of deciding
reducibility of ASD's to be $\mathcal{NP}$-complete, whereas deciding
equivalence of ASD's is shown to be at least as difficult as deciding the
isomorphism of graphs. Furthermore, the latter problem is unlikely to
be~$\mathcal{NP}$-complete, as its $\mathcal{NP}$-completeness would imply
a collapse of the polynomial hierarchy.

The last section (Section~\ref{sec:products}) addresses the direct
product factorization of ASD's. We prove that every device admits a
unique factorization in terms of binary devices, if such a
factorization exists. This result can be seen as a first step towards
answering the general question of the existence of unique
factorizations into (prime) ASD's, which we state as an open problem.

Relevant basic facts about set partitions and the partition lattice
are briefly reviewed in Section~\ref{sec:preliminaries}.

\section{Preliminaries}
\label{sec:preliminaries}

Throughout this paper, we make use of capital calligraphic letters to
denote sets.
An (undirected) \emph{graph} is an ordered pair $\set{G} = (\set{V},
\set{E})$, where $\set{V}$ is the set of \emph{vertices}, and $\set{E}
\subseteq \binom{\set{V}}{2}$ is the set of \emph{edges}~of~$\set{G}$.

A \emph{(set) partition~$\pi$} of a set~$\set{S}$ is a
family~$\{\set{B}_1, \ldots, \set{B}_{k}\}$ of disjoint subsets of
$\set{S}$, called \emph{blocks}, with the property that~$\bigcup_{i =
1}^k \set{B}_i = \set{S}$. We write $s \equiv_\pi t$ whenever both
elements~$s,t \in \set{S}$ are in the same block of~$\pi$. Moreover,
we denote by~$\Part{\set{S}}$ the set of partitions of~$\set{S}$. We
say that~$\pi \in \Part{\set{S}}$ \emph{refines}~$\pi' \in
\Part{\set{S}}$, denoted $\pi \refines \pi'$, if for all $\set{B} \in
\pi$ there exists a $\set{B}' \in \pi'$ such that $\set{B} \subseteq
\set{B}'$. Recall that $(\Part{\set{S}}; \refines)$ is a bounded
lattice (cf.\ e.g.\ \cite{Gra78}), with the minimal element being
$\id_{\set{S}} = \{\{s\} \,|\, s \in \set{S}\}$ and the maximal
element being $\{\set{S}\}$.  The \emph{meet}~of~$\pi, \pi' \in
\Part{\set{S}}$ is the partition~$\pi \land \pi' = \left\{ \set{B}
\cap \set{B'} \,|\, \set{B} \in \pi, \set{B'} \in \pi', \set{B} \cap
\set{B'} \ne \emptyset \right\}$, whereas their \emph{join}~$\pi \lor
\pi'$ is such that~$x \equiv_{\pi \lor \pi'} y$ if and only if we can
find a sequence of elements $x = x_0, x_1, \ldots, x_r = y$ (for some
$r$) such that $x_i \equiv_{\pi} x_{i + 1}$ or $x_i \equiv_{\pi'} x_{i
+ 1}$ holds for all $i = 0, \ldots, {r-1}$. For a set~$\Pi$ of
partitions, we generally write~$\bigwedge \Pi = \bigwedge_{\pi \in
\Pi} \pi$ and $\bigvee \Pi = \bigvee_{\pi \in \Pi} \pi$. Also, such a
set~$\Pi$ is called an \emph{antichain} if $\pi \not\refines \pi'$ for
all distinct~$\pi, \pi' \in \Pi$.

The \emph{direct product} of the partitions~$\pi \in \Part{\set{S}}$ and
$\pi' \in
\Part{\set{S}'}$ is the partition~$\pi \times \pi' =
\left\{\set{B} \times \set{B}' \,|\, \set{B} \in \pi, \set{B}' \in
\pi' \right\} \in \Part{\set{S} \times \set{S}'}$. In particular, we
have~$(s,s') \equiv_{\pi \times \pi'} (t, t')$ if and only if $s
\equiv_{\pi} s'$ and $t \equiv_{\pi'} t'$ for all $s,t \in \set{S}$,
$s',t' \in \set{S}'$.
Let now $\pi, \rho \in \Part{\set{S}}, \pi', \rho' \in
\Part{\set{S}'}$ be partitions. Then, both equalities~$(\pi \land \rho) \times
(\pi' \land \rho') = (\pi \times \pi') \land (\rho \times \rho')$
and~$(\pi \lor \rho) \times (\pi' \lor \rho') = (\pi \times \pi') \lor
(\rho \times \rho')$ hold. Furthermore,~$\pi \times \pi' \refines \rho
\times \rho'$ is satisfied if and only if $\pi \refines \pi'$ and
$\rho \refines \rho'$. We refer the reader to
Appendix~\ref{app:partitions} for a proof of these facts.

Given sets~$\set{S}, \set{S}'$, a partition $\pi \in \Part{\set{S}'}$, and
some function~$\phi: \set{S} \to \set{S'}$, we define~$\pi \circ \phi \in
\Part{\set{S}}$ as the partition such that $x \equiv_{\pi \circ \phi} y$ if
and only if $\phi(x) \equiv_{\pi} \phi(y)$ for all $x, y \in \set{S}$.
Notice that $(\pi \circ \phi) \land (\pi' \circ \phi) = (\pi \land \pi')
\circ \phi$, and~$(\pi \circ \phi) \lor (\pi' \circ \phi) = (\pi \lor \pi')
\circ \phi$.  Moreover, the~\emph{kernel (partition)} of a function~$f:
\set{X} \to \set{Y}$ is $\ker{f} = \{f^{-1}(\{y\}) \,|\, y \in
\range{f}\}$. Given a further function~$\phi: \set{S} \to \set{X}$, we have
$\ker{f \circ \phi} = \ker{f} \circ \phi$.

Finally, recall that a $k$-variate \emph{lattice polynomial}~$p$ in
the variables~$x_1, \ldots, x_k$ is a formal expression of the form
either~(i) $x_i$~for~$i=1, \ldots, k$, or (ii) one of~$q(x_1, \ldots,
x_k) \land q'(x_1, \ldots, x_k)$ and $q(x_1, \ldots, x_k) \lor q(x_1,
\ldots, x_k)$ for $k$-variate lattice polynomials $q, q'$.  Given
partitions $\pi_1, \ldots, \pi_k$, $\rho_1, \ldots, \rho_k$ such that
$\pi_i \refines \rho_i$ for $i=1, \ldots, k$, then $p(\pi_1, \ldots,
\pi_k) \refines p(\rho_1, \ldots, \rho_k)$ holds for every $k$-variate
lattice polynomial~$p$.

\section{Abstract Storage Devices}
\label{sec:model}

\subsection{Definition}
\label{sec:asd_definition}
In the following, we look at storage devices used by two entities, called
the \emph{writer} and the \emph{reader}, respectively\footnote{These
  entities are not necessarily distinct in a physical sense.}. The writer
writes to such a device by selecting a state~$s$ from the~\emph{state
  space} of the device. The reader subsequently chooses a (possibly
randomized) function~$\func$ mapping states to output symbols from a set of
possible such mappings, and obtains the output~$\func(s)$. Note, however,
that the actual labeling of the outputs is irrelevant, as long as the
reader knows a complete description of the function to be read out.  In
particular, as we only focus on devices whose behavior is entirely
\emph{deterministic}, we abstract from the notion of an output domain and
we solely describe the kernel partitions of the functions of the storage
device.  This allows us to formulate the following combinatorial
abstraction of deterministic devices.

\begin{definition}
  An \emph{abstract storage device (ASD)}~$D$ is a pair~$D =
  \left(\Stateof{D}, \Partof{D} \right)$, where~$\Stateof{D}$ is a set
  called the \emph{state space of $D$}, and $\Partof{D}$ is a family
  of partitions of~$\Stateof{D}$, called the \emph{partition set
  of $D$}.
\end{definition}

For an ASD $D$, a \emph{write operation} of the writer consists in
selecting a state~$\state \in \Stateof{D}$, and in a
subsequent~\emph{read operation} the reader selects a partition~$\pi
\in \Partof{D}$ and learns the (unique) block~$\set{B} \in \pi$ such
that $\state \in \set{B}$. We assume that a single read operation is
performed. Furthermore, in the following, we are going to focus on
ASD's with finite state space and partition set.

Whenever~$\id_{\Stateof{D}} \in \Partof{D}$, the reader can
distinguish any pair of states with a single read operation. In this
case, $D$ is called \emph{perfect}, and it is called
\emph{non-perfect} otherwise. If the partition set contains only the
trivial partition~$\{\Stateof{D}\}$, the ASD is
called~\emph{trivial}. Moreover, it is called~\emph{$r$-regular} if
$|\pi| = r$ for all $\pi \in \Partof{D}$. In particular, $2$-regular
ASD's are also called \emph{binary}.

The following are examples of ASD's.

\paragraph{Perfect device.}  For a given set $\set{X}$,
the ASD~$C_{\set{X}}$ has state space $\set{X}$
and its state can be retrieved perfectly, that is, $\Partof{D} = \{
\id_{\set{X}}\}$.  The special case where $\set{X} = \{1, \ldots,
m\}$ for $m \in \mathbb{N}$ is denoted as $C_m$.

\paragraph{Projective device.}  
For $i \in \{1, \ldots, n\}$, we denote by~$p_i:\{0,1\}^n \to \{0,1\}$ the
function such that~$p_i(x_1, \ldots, x_n) = x_i$ for all $(x_1, \ldots,
x_n) \in \{0,1\}^n$.  The \emph{projective device}~$P_n$ has state space
$\Stateof{P_n} = \{0,1\}^n$ and its partition set is $\Partof{P_n} = \{
\ker{p_i} \,|\, i=1, \ldots, n\}$. This device is similar to the
\emph{$1$-out-of-$n$ oblivious transfer (OT)} primitive considered in
cryptography~(introduced in~\cite{Rab81}). One may also extend this device
to allow for retrieving any~$k < n$ consecutive bits of the state. Such a
device could be used to model a tape-based storage device.

\paragraph{Linear device.} The \emph{linear device}
$L_{n,k}$ where $n \ge k$ is the ASD having state space
$\Stateof{L_{n,k}} = \{0,1\}^n$, and the partition set is the set of
the kernel partitions of all linear maps $\{0,1\}^n \to \{0,1\}^k$.
We denote by $L_n$ the binary ASD~$L_{n,1}$.
\vspace*{1em}

One way of constructing a complex device from simpler devices is the
parallel composition of two ASD's to obtain a new ASD modeling a
setting where the reader and the writer use both devices in a
\emph{non-adaptive} fashion. That is, if $D$ has state $s$ and $D'$
has state $s'$, the reader first selects {\em both} partitions~$\pi
\in \Partof{D}$ and $\pi' \in \Partof{D'}$, and only subsequently
learns the unique blocks~$\set{B} \in \pi$, $\set{B}' \in \pi'$ such
that $s \in \set{B}$ and $s' \in \set{B}'$.

\begin{definition}
The \emph{direct product}~$D \times D'$ of the ASD's $D, D'$ is the
ASD with $\Stateof{D \times D'} = \Stateof{D} \times \Stateof{D'}$ and
$\Partof{D \times D'} = \{ \pi \times \pi' \, |\, \pi \in \Partof{D},
\pi' \in \Partof{D'} \}$.
\end{definition}

For example, since~$\id_{\Stateof{D} \times \Stateof{D'}} = \pi \times
\pi'$ holds if and only if $\pi = \id_{\Stateof{D}}$ and $\pi' =
\id_{\Stateof{D'}}$, we immediately see that~$D \times D'$ is perfect
if and only if both $D$ and $D'$ are perfect.

In general, we may want to look at more than a single read operation.
For an integer~$k \ge 1$ and an ASD~$D$, we denote as~$D^{(k)}$ the
ASD with~$\Stateof{D^{(k)}} = \Stateof{D}$ and~$\Partof{D^{(k)}} =
\Bigl\{\bigwedge_{i = 1}^k \pi_i \,\Big|\, \pi_i \in \Partof{D}, i =
1, \ldots, k \Bigr\}$.  It models the scenario where the reader is
allowed to perform (at most) $k$ non-adaptive read operations, i.e.\
given state~$s \in \Stateof{D}$, it first chooses $k$
partitions~$\pi_1, \ldots, \pi_k \in \Partof{D}$ to be retrieved, and
only subsequently learns the corresponding blocks~$\set{B}_1 \in
\pi_1, \ldots, \set{B}_k \in \pi_k$ such that $s \in \bigcap_{i = 1}^k
\set{B}_i$.

Note that both the direct product and the device~$D^{(k)}$ can be
extended to allow for adaptive read operations, as it essentially
suffices to consider all partitions induced by every possible
(deterministic) retrieval strategy. However, we do not address this
case in this paper.

\subsection{Reducibility and Equivalence}
\label{sec:reducibility}

In the problem of reducibility of ASD's, we want to decide whether an
ASD $D$ can be implemented by a second ASD $D'$. This is formalized by
the following definition.

\begin{definition}
  \label{def:reduction} We say that an ASD~$D$ is \emph{reducible} to
  an ASD~$D'$, denoted $D \le D'$, if there exist functions $\phi:
  \Stateof{D} \to \Stateof{D'}$ and $\alpha: \Partof{D} \to
  \Partof{D'}$ such that~$\alpha(\pi) \circ \phi \refines \pi$ for
  all~$\pi \in \Partof{D}$. Such a pair of functions~$(\phi, \alpha)$
  is called a \emph{reduction of $D$ to $D'$}.
\end{definition}

In order to clarify this concept, consider the following abstraction
in terms of ASD's of the game introduced in
Section~\ref{sec:motivation}. The writer and the reader are given an
ASD~$D'$ as well as the description of a further ASD~$D$. The writer
is told an arbitrary state~$\state \in \Stateof{D}$ and selects the
state~$\phi(\state) \in \Stateof{D'}$ for~$D'$. Later, an arbitrary
partition~$\pi \in \Partof{D}$ is revealed to the reader, and it
performs a read operation for a partition~$\alpha(\pi) \in
\Partof{D'}$. The goal is to find appropriate functions~$\phi:
\Stateof{D} \to \Stateof{D'}$ and~$\alpha: \Partof{D} \to \Partof{D'}$
such the reader can \emph{perfectly} guess the unique block~$\set{B}
\in \pi$ such that $\state \in \set{B}$ from the result of
retrieving~$\alpha(\pi)$ from~$D'$. If such functions exist, the
writer and the reader can simulate~$D$ using $D'$. Note that the
ASD~$D$ itself can alternatively be seen as the specification of a
particular game the writer and the reader try to win by using the
ASD~$D'$.

It is easy to see that the condition~$\alpha(\pi) \circ \phi \refines
\pi$ must hold. Otherwise, there would be~$\state, \state' \in
\Stateof{D}$ such that $\state \notequiv_\pi \state'$, but
$\phi(\state) \equiv_{\alpha(\pi)} \phi(\state')$, and hence~$\state$
and $\state'$ could not be distinguished. Conversely, if~$\alpha(\pi)
\circ \phi \refines \pi$, then given state~$\state \in \Stateof{D}$
and $\set{B}' \in \alpha(\pi)$ such that $\phi(\state) \in \set{B}'$,
there exists a unique block~$\set{B} \in \pi$ such that $\state \in
\set{B}$. Hence, Definition~\ref{def:reduction} expresses the precise
condition in order for $\phi$ and $\alpha$ to be a winning strategy in
the game.

Reducibility is a reflexive and transitive relation. However, it is
not antisymmetric, and thus it is only a \emph{quasi-order} on the set
of ASD's. In this respect, we say that two ASD's~$D, D'$ are
\emph{equivalent}, denoted $D \equiv D'$, if both $D \le D'$ and $D'
\le D$ hold. The relation~$\equiv$ is an equivalence relation and
reducibility implicitly defines a partial order on its equivalence
classes.

The following proposition relates reducibility to direct products and
multiple read operations.

\begin{proposition} 
\label{prop:red_prod_seq}
  Let $D,D',E,E'$ be ASD's.
        \begin{enumerate}[(i)]  \item If $D \le D'$ and $E \le E'$,
        then $D \times E \le D' \times E'$. 
      \item If $D \le D'$, then $D^{(k)} \le
        D'^{(k)}$. \end{enumerate}
\end{proposition}
\begin{proof}
  The first claim is obvious. For the second one, let $(\phi, \alpha)$ be a
  reduction of~$D$ to~$D'$.  Define~$\tilde{\alpha}: \Partof{D^{(k)}} \to
  \Partof{D'^{(k)}}$ such that $\tilde{\alpha}(\bigwedge_{i = 1}^k \pi_i) =
  \bigwedge_{i = 1}^k \alpha(\pi_i)$. Then, $(\phi, \tilde{\alpha})$
  reduces $D^{(k)}$ to $D^{(k')}$, since~$\tilde{\alpha}(\bigwedge_{i =
    1}^k \pi_i) \circ \phi = \left( \bigwedge_{i = 1}^k \alpha(\pi_i)
  \right) \circ \phi = \bigwedge_{i = 1}^k (\alpha(\pi_i) \circ \phi)
  \refines \bigwedge_{i = 1}^k \pi_i$. 
\end{proof}

The perhaps most natural question related to storage devices is to
determine how many bits of information can be reliably stored in it
with the guarantee of no errors at read out. This quantity can be
expressed in terms of the largest perfect device that can be reduced
to the considered device.

\begin{definition}
The \emph{storage capacity} of an ASD~$D$ is $C(D) = \max \{ \log m
\,|\, C_m \le D, m \in \mathbb{N} \}$.
\end{definition}

Equivalence of ASD's captures that two ASD's $D$ and $D'$ such that $D
\equiv D'$ have the same behavior. As an example, it is clear that~$D
\times D' \equiv D' \times D$, and that~ $D \times (D' \times D'')
\equiv (D \times D') \times D''$, that is, the direct product is
commutative and associative with respect to equivalence. The direct
product of $D_1, \ldots, D_n$ is thus simply written as~$\bigtimes_{i
= 1}^n D_i$, and $D^k = \bigtimes_{i = 1}^k D$ for any
device~$D$. Finally, notice that $D \times E \equiv D$ holds for any
trivial device~$E$.

\subsection{Minimality}
\label{sec:minimality}

In this section, we have a closer look at the equivalence
relation~$\equiv$ and at the inner structure of its equivalence
classes. In particular, we are interested in the minimal number of
states and partitions needed in order to implement the functionality
of a certain ASD.

\begin{definition}
An ASD~$D$ is \emph{state-minimal} if there is no equivalent device
$D'$ with $|\Stateof{D'}| < |\Stateof{D}|$. Furthermore,~$D$ is
\emph{partition-minimal} if there is no equivalent device~$D'$ with
$|\Partof{D'}| < |\Partof{D}|$. Finally, we say that~$D$ is
\emph{minimal} if $D$ is both state and partition-minimal.
\end{definition}

For every ASD $D$ there exist by definition equivalent ASD's~$D'$ and
$D''$ such that $D'$ is state-minimal and $D''$ is partition
minimal. However, it is not clear whether an equivalent ASD exists
that satisfies both, i.e., which is minimal. This is shown in the
following theorem, which also provides an equivalent characterization
of state and partition-minimality.

\begin{theorem}
\label{thm:minimality}
  For an ASD $D$ we have the following.
  \begin{enumerate}[(i)]
  \item $D$ is state-minimal if and only if for all pairs of distinct
    states $\state, \state' \in \Stateof{D}$ there exists a set
    partition~$\pi \in \Partof{D}$ such that~$\state \notequiv_{\pi}
    \state'$. In particular, this holds if and only if $\bigwedge
    \Partof{D} = \id_{\Stateof{D}}$.
  \item $D$ is partition-minimal if and only if $\Partof{D}$ is an
    antichain (with respect to $\refines$).
\end{enumerate}
Furthermore, for every ASD~$D$, there exists a minimal
ASD~$D' \equiv D$.
\end{theorem}
\begin{proof}
We prove the two parts of the theorem separately.
\begin{enumerate}[(i)]
\item Assume that $D$ is a state-minimal ASD and that there are
  distinct states $\state_1, \state_2 \in \Stateof{D}$ such that for
  all $\pi \in \Partof{D}$ we have $\state_1 \equiv_{\pi}
  \state_2$. Construct a new ASD~$D'$ as follows. We define
  $\Stateof{D'} := \Stateof{D} - \{\state_2\}$ and $\Partof{D'} :=
  \{\pi \circ \psi \,|\, \pi \in \Partof{D} \}$ where~$\psi:
  \Stateof{D'} \to \Stateof{D}$ is such that $\psi(\state) = \state$.
  Clearly, $D' \le D$. On the other hand, one can easily see that $D
  \le D'$: Define a function $\phi: \Stateof{D} \to \Stateof{D'}$ as
  \begin{displaymath}
    \phi(\state) := \left\{
      \begin{array}{ll}
      \state, & \text{ if } \state \in \Stateof{D'}, \\ \state_1, &
      \text{ if } \state = \state_2,
      \end{array}
      \right.
  \end{displaymath}
  and let $\alpha$ be such that~$\alpha(\pi) = \pi \circ
  \psi$. Then~$(\phi, \alpha)$ is a reduction of $D$ to $D'$
  as~$\alpha(\pi) \circ \phi = \pi \circ (\psi \circ \phi) \refines
  \pi$ because of the choice of $\state_1$ and $\state_2$.
    
  For the converse, assume that for an ASD $D$ we have for every pair
  of distinct states $\state, \state' \in \Stateof{D}$ a partition
  $\pi \in \Partof{D}$ such that $\state \notequiv_{\pi} \state'$.
  Assume now that $D$ is not state-minimal. That is, there is a
  device~$D'$ with $|\Stateof{D'}| < |\Stateof{D}|$ and $D' \equiv D$.
  Let $(\phi, \alpha)$ be a reduction of $D$ to $D'$. There must be
  two states $\state_1, \state_2 \in \Stateof{D}$ such that
  $\phi(\state_1) = \phi(\state_2)$, and hence for all $\pi' \in
  \Partof{D'}$ we have $\phi(\state_1) \equiv_{\pi'}
  \phi(\state_2)$. In particular, let $\pi \in \Partof{D}$ be such
  that $\state_1 \notequiv_{\pi} \state_2$. Then $\pi' \circ \phi
  \not\refines \pi$ for all $\pi' \in \Partof{D'}$, and thus $D \nleq
  D'$.

  It is straightforward to verify that $\bigwedge \Partof{D} =
  \id_{\Stateof{D}}$ holds if and only if for all $\state, \state' \in
  \Stateof{D}$ there exists $\pi \in \Partof{D}$ such that~$\state
  \not\equiv_{\pi} \state'$.

\item Assume that $D$ is a partition-minimal ASD and that $\Partof{D}$
  is not an antichain. That is, there exist distinct~$\pi_1, \pi_2 \in
  \Partof{D}$ such that $\pi_1 \refines \pi_2$. We build a new device
  $D'$ with $\Stateof{D'} := \Stateof{D}$ and $\Partof{D'} :=
  \Partof{D} - \{\pi_2\}$. Clearly, we have $D' \le D$. Furthermore,
  define $\phi: \Stateof{D} \to \Stateof{D'}$ as the identity and
  $\alpha: \Partof{D} \to \Partof{D'}$ such that
  \begin{displaymath}
  \alpha(\pi) := \left\{
      \begin{array}{ll}
      \pi, & \text{ if } \pi \in \Partof{D'}, \\ \pi_1, & \text{ if }
      \pi = \pi_2,
      \end{array}
      \right.
  \end{displaymath}
  for all $\pi \in \Partof{D}$. This implies that $\alpha(\pi) =
  \alpha(\pi) \circ \phi \refines \pi$ for all $\pi \in \Partof{D}$,
  and thus $D \le D'$.  Consequently, $D' \equiv D$.  However,
  $|\Partof{D'}| = |\Partof{D}| - 1$, which contradicts the fact that
  $D$ is partition-minimal.
  
  For the converse, assume that~$\Partof{D}$ is an antichain. Without
  loss of generality let $D$ be state-minimal.  Towards a
  contradiction, additionally assume that $D$ is not partition
  minimal, that is, there is $D'$ such that $D' \equiv D$ and
  $|\Stateof{D}| = |\Stateof{D'}|$ but $|\Partof{D'}| < |\Partof{D}|$.
  In particular, let~$(\phi', \alpha')$ and~$(\phi'', \alpha'')$ be
  reductions of $D$ to $D'$ and of $D'$ to $D$, respectively. Note
  that $|\range{\alpha'}| \le |\Partof{D'}| < |\Partof{D}|$ by our
  assumption. Moreover, let $\phi:= \phi'' \circ \phi'$ and $\alpha :=
  \alpha'' \circ \alpha'$.  Then,~$(\phi, \alpha)$ is a reduction of
  $D$ to itself where the function $\alpha$ is not injective, since
  $|\range{\alpha}| \le |\range{\alpha'}| < |\Partof{D}|$. Moreover,
  as $D$ is state-minimal, $\phi$ is a permutation of
  $\Stateof{D}$. (Otherwise, one would easily be able to build an
  equivalent ASD with fewer states, hence contradicting
  state-minimality.)  Since~$\alpha$ is not injective, there are
  distinct~$\pi_1, \pi_2 \in \Partof{D}$ such that $\alpha(\pi_1) =
  \alpha(\pi_2)$.  Additionally, we have $\alpha(\pi_1) \circ \phi
  \refines \pi_1$ as well as $\alpha(\pi_1) \circ \phi = \alpha(\pi_2)
  \circ \phi \refines \pi_2$, and therefore $\alpha(\pi_1) \circ \phi
  \refines \pi_1 \land \pi_2$.  Also, since $\alpha$ maps partitions
  of $D$ to partitions of $D$, for all integers $k \ge 1$, we have
  \begin{equation} 
    \label{eq:pfmin1} 
    \alpha^k(\pi_1) \circ \phi^k \refines \pi_1 \land
    \pi_2.
  \end{equation}
  Because of our assumption, $\{\pi_1, \pi_2\}$ is an antichain, and
  therefore, $\pi_1 \land \pi_2 \notin \{\pi_1, \pi_2\}$, which
  implies $\pi_1 \land \pi_2 \srefines \pi_1$ and $\pi_1 \land \pi_2
  \srefines \pi_2$. Using this fact, for all integers $k \ge 1$, we
  see that $\alpha^k(\pi_1) \notin \{\pi_1, \pi_2 \}$ since 
  \begin{displaymath}
    |\alpha^k(\pi_1)| = |\alpha^k(\pi_1) \circ \phi^k| \ge
    |\pi_1 \land \pi_2| > \max\{|\pi_1|, |\pi_2|\}.
  \end{displaymath}
  However, there has to exist an integer $k'$ such that $\phi^{k'}$ is
  the identity permutation.  By plugging $k'$ into (\ref{eq:pfmin1})
  we obtain
  \begin{displaymath}
    \alpha^{k'}(\pi_1) \refines \pi_1 \land \pi_2
    \srefines \pi_1,
  \end{displaymath}
  which contradicts the fact that $\Partof{D}$ is an antichain.
\end{enumerate}

Note that by the proofs of (i) and (ii) we see that, given an ASD~$D$,
one can iteratively construct a state-minimal ASD~$D'$ such that $D'
\equiv D$. Furthermore, one can construct out of~$D'$ a
partition-minimal ASD~$D'' \equiv D' \equiv D$ such
that~$|\Stateof{D'}| = |\Stateof{D''}|$. Hence~$D''$ is minimal, and
this concludes the proof of Theorem~\ref{thm:minimality}.
\end{proof}

As an example, observe that the projective device~$P_n$ is state
minimal. Indeed, given distinct~$x, x' \in \{0,1\}^n$, there exists a
component $i$ such that $x_i \ne x'_i$, and thus $x
\notequiv_{\ker{p_i}} x'$. This also implies that the linear
device~$L_n$ is state-minimal. Furthermore, every $r$-regular
device (for some $r$) is necessarily partition-minimal, since any two
partitions with the same number of blocks are either equal or
incomparable (with respect to $\refines$).

The following lemma provides some properties of minimal devices with
respect to device reducibility.

\begin{lemma}
\label{lem:minimality}
\begin{enumerate}[(i)]
\item 
If $D, D'$ are state-minimal and $(\phi, \alpha)$ reduces $D$ to $D'$,
then $\phi$ is injective. In particular, $|\Stateof{D}| \le
|\Stateof{D'}|$.
\item 
If $D, D'$ are both $r$-regular for some $r$ (and hence partition
minimal) and $(\phi, \alpha)$ reduces $D$ to $D'$, then $\alpha$ is
injective. In particular, $|\Partof{D}| \le |\Partof{D'}|$.
\item 
If $D, D'$ are both state-minimal (partition-minimal), then the direct
product~$D \times D'$ is state-minimal (partition-minimal).
\end{enumerate}
\end{lemma}
\begin{proof} To prove (i), assume that there are indeed $\state_0, \state_1 \in
  \Stateof{D}$ such that $\phi(\state_0) = \phi(\state_1)$, then there
  exists a partition~$\pi \in \Stateof{D}$ such that $s_0 \not\equiv
  s_1$, while for all $\pi' \in \Partof{D'}$ we have $s_0
  \equiv_{\pi' \circ \phi} s_1$ and hence $\pi' \circ \phi
  \not\refines \pi$.
  
  For (ii), assume that $\alpha$ is not an injection, then there exists
  $\pi_0 \ne \pi_1 \in \Partof{D}$ such that $\alpha(\pi_1) =
  \alpha(\pi_2)$. That is $\alpha(\pi_1) \circ \phi \refines \pi_1 \land
  \pi_2$. But then $|\alpha(\pi_1) \circ \phi| \ge |\pi_1 \land \pi_2| >
  r$, since $\Partof{D}$ is an antichain. However, this contradicts the
  fact that~$|\alpha(\pi_1) \circ \phi| \le r$.
  
  Finally, in order to prove (iii), let $D, D'$ be state-minimal. Then
  $\bigwedge \Partof{D \times D'} = \Bigl( \bigwedge \Partof{D}\Bigr)
  \times \Bigl( \bigwedge \Partof{D'} \Bigr) = \id_{\Stateof{D}} \times
  \id_{\Stateof{D'}} = \id_{\Stateof{D \times D'}}$, and thus $D \times D'$
  is state-minimal by Theorem~\ref{thm:minimality}. Furthermore, let $D,
  D'$ be partition-minimal, and assume $D \times D'$ is not. Then there
  exist distinct $\pi \times \pi', \rho \times \rho' \in \Partof{D \times
    D'}$ such that $\pi \times \pi' \refines \rho \times \rho'$. But then
  $\pi \refines \rho$ and $\pi' \refines \rho'$. Since $\pi \ne \rho$ or
  $\pi' \ne \rho'$ holds, at least one of $D$ and $D'$ is not
  partition-minimal.
\end{proof}

It also turns out that equivalence of devices is easier
to characterize in the minimal case. 

\begin{proposition}
  \label{prop:minimal_equivalence}
        Let $D, D'$ be minimal ASD's. Then $D \equiv D'$ if and only
        if there exist bijections~$\phi: \Stateof{D} \to \Stateof{D'}$
        and $\alpha: \Partof{D} \to \Partof{D'}$ such that $\pi =
        \alpha(\pi) \circ \phi$ for all $\pi \in \Partof{D}$, or,
        equivalently, $\pi' = \alpha^{-1}(\pi') \circ \phi^{-1}$ for
        all $\pi' \in \Partof{D'}$.
\end{proposition}
\begin{proof}
  Clearly, if such bijections exist, then $D \equiv D'$. Now, assume
  that $D \equiv D'$, then there exists a reduction $(\phi, \alpha)$
  of $D$ to $D'$. Note that $\phi$ must be a bijection by
  Lemma~\ref{lem:minimality}. Furthermore, $\alpha$ must also be a
  bijection, otherwise there would be an equivalent ASD with fewer
  partitions, contradicting the partition-minimality of~$D$.
  
  Assume towards a contradiction that there is $\pi \in \Partof{D}$
  such that $\alpha(\pi) \circ \phi \srefines \pi$. Note that since
  $D' \le D$, there exists a reduction $(\phi', \alpha')$ of $D'$ to
  $D$ where $\phi'$ and $\alpha'$ are both bijections. Consequently,
  there exists a reduction $(\tilde{\phi}, \tilde{\alpha})$ from $D$
  to itself where $\tilde{\phi} := \phi' \circ \phi$ and
  $\tilde{\alpha} := \alpha' \circ \alpha$ are permutations of
  $\Stateof{D}$ and $\Partof{D}$, respectively.  Moreover, for all $k
  \ge 1$, we have $\tilde{\alpha}^k(\pi) \circ \tilde{\phi}^k \refines
  \tilde{\alpha}(\pi) \circ \tilde{\phi} \refines \alpha(\pi) \circ
  \phi \srefines \pi$. Thus, by choosing $k \ge 1$ such that
  $\tilde{\phi}^k$ is the identity permutation, we obtain a
  contradiction to the partition-minimality of~$D$.
\end{proof}

For example, given ASD's $D, D'$, where $\Partof{D} = \{ \pi_1,
\ldots, \pi_k\}$, as well as a $k$-variate lattice polynomial~$p$,
Proposition~\ref{prop:minimal_equivalence} implies that
$p(\alpha(\pi_1) \circ \phi, \ldots, \alpha(\pi_k) \circ \phi) =
p(\alpha(\pi_1), \ldots, \alpha(\pi_k)) \circ \phi = p(\pi_1, \ldots,
\pi_k)$. As $\phi$ is a bijection, in order to prove that $D
\not\equiv D'$ it is sufficient to find a $k$-variate lattice
polynomial~$p$ such that $|p(\pi_1, \ldots, \pi_k)| \ne
|p(\alpha(\pi_1), \ldots, \alpha(\pi_k)|$.

\subsection{Necessary Conditions for Reducibility}
\label{sec:orderpreserving}

In this section, we discuss easily characterizable necessary conditions for
reducibility.  Let~$\D$ be a set of ASD's and let $f: \D \to \mathbb{R}$ be
a function.  We say that~$f$ is \emph{order-preserving on $\D$} if $D \le
D'$ implies $f(D) \le f(D')$ for all ASD's~$D, D' \in \D$.  In particular,
note that~$f(D) = f(D')$ whenever~$D \equiv D'$.  Such a function yields a
necessary condition for reducibility.  In the following paragraphs, we
discuss three order-preserving functions.

\paragraph{Storage capacity.} 
The storage capacity (cf.\ Section~\ref{sec:reducibility}) is
order-preserving on the set of all ASD's: Given $D,D'$ such that $D
\le D'$, let $m$ be maximal such that $C_m \le D$. By transitivity we
have $C_m \le D'$, and hence~$\log m = C(D) \le C(D')$. The storage
capacity is easy to compute, as stated in the following proposition,
which also provides properties with respect to direct products and
multiple read operations.

\begin{proposition} 
  \label{prop:cap} 
  \begin{enumerate}[(i)]
    \item $C(D) = \max_{\pi \in \Partof{D}} \log |\pi|$ for all ASD's
    $D$.
  \item \label{it:detcap2} $C(D \times D') = C(D) + C(D')$ for all
  ASD's $D, D'$.
  \item \label{it:detcap3} For all $k \geq 1$, we have $C(D^{(k)})
    \leq k \cdot C(D)$ for all ASD's $D$.
  \end{enumerate}
\end{proposition}

The first claim follows from the simple observation that~$C_m \le D$
holds if and only if there exists~$\pi \in \Partof{D}$ such that
$|\pi| \ge m$. The simple proofs of (ii) and (iii) are omitted.

For instance, $C(D) = \log r$ for every $r$-regular
ASD~$D$. Furthermore, the storage capacity allows us to easily see
that~$L_2 \times L_2 \times L_2 \nleq L_3 \times L_3$, since $C(L_2
\times L_2 \times L_2) = 3\cdot C(L_2) = 3 $, but $C(L_3 \times L_3) =
2\cdot C(L_3) = 2$.

\paragraph{State complexity.} 
The \emph{state complexity}~$\sigma(D)$ of an ASD $D$ provides the
minimal number of states that are necessary in order to reproduce the
behavior of $D$, that is, $\sigma(D) = \min_{E \equiv D} \log
|\Stateof{E}|$. The state complexity is order-preserving: Given
devices~$D, D'$, let $E,E'$ be state-minimal such that~$D \equiv E$
and $D' \equiv E'$. Since $D \le D'$, we have $E \le E'$ by
transitivity, and by Lemma~\ref{lem:minimality} this implies
$\sigma(D) = \log |\Stateof{E}| \le \log |\Stateof{E'}| =
\sigma(D')$. Furthermore, $\sigma(D \times D') = \sigma(D) +
\sigma(D')$ by Lemma~\ref{lem:minimality}.

Note that $D \le C_{2^{\sigma(D)}}$, whereas
Lemma~\ref{lem:minimality} yields $D \nleq C_{m'}$ for all $m' <
2^{\sigma(D)}$. For this reason, we obtain~$\sigma(D) = \min \{ \log m
\,|\, m \in \mathbb{N}, D \le C_m\}$. Therefore, the state
complexity~$\sigma(D)$ provides the minimal amount of reliable storage
in terms of bits needed to win the game (in the sense of
Section~\ref{sec:reducibility}) described by the ASD~$D$.

\paragraph{Perfectness index.} 
The \emph{perfectness index}~$i(D)$ of a device~$D$ is the minimal
integer~$k$ such that~$D^{(k)}$ is perfect, if such $k$
exists. Otherwise, $i(D) = \infty$. Thus, $i(D)$ provides the minimal
number of read operations needed to retrieve the state
perfectly. If~$i(D)$ is finite, then in particular~$i(D) \le
|\Partof{D}|$, and by Theorem~\ref{thm:minimality} $i(D)$ is bounded
if and only if $D$ is state-minimal. In the following, for an
integer~$m$, consider the set of ASD's~$\D_m$ such that for all $D \in
\D_m$ we have $|\Stateof{D}| = m$.

\begin{proposition}~\label{prop:imperf} 
  Let $D, D' \in \D_m$ for some $m$ be such that $D \le D'$.  Then,
  $i(D) \ge i(D')$. That is, $D \mapsto -i(D)$ is an order-preserving
  function on~$\D_m$.
\end{proposition}
\begin{proof}
  If $i(D) = \infty$ holds, the claim is trivially
  satisfied. Therefore, assume that $i(D)$ is finite, and, towards a
  contradiction, that $D \le D'$, but $i(D) < i(D')$. There is an
  integer $k \ge 1$ such that $D^{(k)}$ is perfect, but $D'^{(k)}$ is
  not. Thus, $\id_{\Stateof{D}} \in \Partof{D^{(k)}}$, but
  $\id_{\Stateof{D'}} \notin \Partof{D'^{(k)}}$. Since $|\Stateof{D}|
  = |\Stateof{D'}| = m$, for all possible~$\phi: \Stateof{D} \to
  \Stateof{D'}$ there is no partition~$\pi' \in \Partof{D'^{(k)}}$
  such that $\pi' \circ \phi \refines
  \id_{\Stateof{D}}$. Hence~$D^{(k)} \nleq D'^{(k)}$, which
  contradicts~$D \le D'$ according to
  Proposition~\ref{prop:red_prod_seq}. 
\end{proof}

One can easily verify that $i(\bigtimes_{i = 1}^n D_i) = \max_{1 \le i
\le n} i(D_i)$ for any ASD's $D_1, \ldots, D_n$. Furthermore, $i(L_n) =
n$, since exactly~$n$ distinct, linearly independent, linear
predicates have to be read out to learn the state. 
As an example, consider the ASD's~$L_4 \times L_2$ and $L_3
\times L_3$. By the above, we have~$i(L_4 \times L_2) = 4$, and $i(L_3
\times L_3) = 3$.  Therefore, $L_3 \times L_3 \nleq L_4 \times L_2$ by
Proposition~\ref{prop:imperf}.

\vspace*{1em} The presented quantities are related by the following
proposition.

\begin{proposition}
For all ASD's~$D$, we have $\sigma(D) \le i(D) \cdot C(D)$.
\end{proposition}
\begin{proof}
The claim is trivially true if $i(D) = \infty$. Otherwise, we just
combine the facts that~$\sigma(D) \le C(D^{(i(D))}) =
\log|\Stateof{D}|$ and that~$C(D^{(i(D))}) \le i(D) \cdot C(D)$. 
\end{proof}

\section{Complexity of Reducibility and Equivalence}
\label{sec:complexity}

We investigate the computational complexity of deciding reducibility
and equivalence of ASD's. Both problems are obviously in
$\mathcal{NP}$, since given a reduction~$(\phi, \alpha)$ reducibility
can be verified in polynomial-time (in the numbers of states and
partitions)\footnote{We assume some canonical encoding of ASD's.}, and
hence also equivalence (by giving two corresponding reductions). In
this section, we prove the following theorem.

\begin{theorem}
  \label{thm:npcompleteness} Reducibility of ASD's is
  $\mathcal{NP}$-complete. Furthermore, deciding equivalence of ASD's
  is at least as hard as deciding graph isomorphism.
\end{theorem}

First, we briefly recall some graph-theoretic notions. A graph~$\G =
(\V, \E)$ is \emph{isomorphic} to $\G' = (\V', \E')$, denoted~$\G \cong
\G'$, if there exists a bijection~$\phi: \V \to \V'$ such that $\{v,
w\} \in \E$ if and only if $\{\phi(v), \phi(w) \} \in \E'$.
Furthermore, $\G$ is a \emph{subgraph} of $\G'$ if $\V \subseteq \V'$
and $\E \subseteq \E'$.  Finally, $\G$ is \emph{contained} in $\G'$,
denoted~$\G \preceq \G'$, if there exists a subgraph~$\set{H}$ of $\G$
such that $\G \cong \set{H}$. Let~$\set{K}_k$ be the complete graph on
$k$ vertices. The \emph{$k$-clique problem} consists in deciding,
given a graph~$\G$, whether $\set{K}_k \preceq \G$. For arbitrary~$k$,
this is a well-known $\mathcal{NP}$-complete problem. 

In order to prove Theorem~\ref{thm:npcompleteness}, we introduce a class of
ASD's representing graphs.  For a given graph~$\set{G} = (\V, \E)$, we
define its \emph{graph device}~$D(\G)$ as the $3$-regular ASD such
that~$\Stateof{D(\G)} = \V$ and $\Partof{D(\G)} = \{ \pi_e \,|\, e \in
\E\}$, where for $e = \{u, v\} \in \E$, we have $\pi_e = \bigl\{\{u\},
\{v\}, V - \{u, v\} \bigr\}$. Note that graph devices are only meaningful
if $|\V| \ge 4$, since in the case where $|\V| = 3$, all edges define the
same partition.

For instance, if one takes the complete graph~$\set{K}_k$ (for $k \ge 4$),
the resulting graph device~$D(\set{K}_k)$ has state space~$\{1, \ldots,
k\}$ and all its partitions are of the form $\{\{i\}, \{j\}, \{1, \ldots,
k\} - \{i,j\}\}$ for all~$i < j$, $i, j \in \{1, \ldots, k\}$.

The following result can easily be verified
using~Theorem~\ref{thm:minimality}.

\begin{lemma}
  \label{lem:graphmin}
  The ASD~$D(\G)$ is minimal for all graphs $\G = (\V, \E)$ with $|\V|
  \ge 4$ and no isolated\footnote{A vertex~$v \in \V$ is
  \emph{isolated} if there exists no $e \in \E$ such that $v \in e$.}
  vertices.
\end{lemma}

The following lemma is the central point in the proof of
Theorem~\ref{thm:npcompleteness}.

\begin{lemma} 
\label{lem:red2}
  Let $\G = (\V,\E)$ and $\G' = (\V', \E')$ be graphs with no isolated
  vertices such that $\min \{|\V|, |\V'|\} \ge 4$. Then, $\G \preceq
  \G'$ if and only if $D(\G) \le D(\G')$.
\end{lemma}

\begin{proof}
  For notational convenience, let~$\Partof{D(\G)} = \{\pi_e \,|\, e
  \in \E\}$ and~$\Partof{D(\G')} = \{\pi'_{e'} \,|\, e' \in \E'\}$.
  If $\G \preceq \G'$, then there is an injective map $\phi: \V \to
  \V'$ such that, for all $u, v \in \V$, $\{u,v\} \in \E$ implies
  $\{\phi(u), \phi(v)\} \in \E'$. That is, for all $e \in \E$, we have
  $\pi'_{\phi(e)} \in \Partof{D(\G')}$.  Construct a map $\alpha:
  \Partof{D(\G)} \to \Partof{D(\G')}$ such that for all $e \in \E$, we
  set $\alpha(\pi_e) = \pi'_{\phi(e)}$. One can now easily see that
  for all $e \in \E$, we have $\pi_e = \pi'_{\phi(e)} \circ \phi$, and
  thus $(\phi, \alpha)$ reduces $D(\G)$ to $D(\G')$.
  
  For the converse, assume that $D(\G) \le D(\G')$, and let~$(\phi,
  \alpha)$ be a reduction of $D(\G)$ to $D(\G')$. Since both graphs
  have at least four vertices $D(\G)$ and $D(\G')$ are both
  state-minimal by Lemma~\ref{lem:graphmin}, and therefore the
  function $\phi$ is injective by Lemma~\ref{lem:minimality}. For all
  $e \in \E$, there is $e' \in \E'$ such that $\alpha(\pi_e) =
  \pi_{e'}$ and such that $\pi_e = \pi'_{e'} \circ \phi$. For all~$e =
  \{v,w\}$, this means that $\phi(v) \notequiv_{\pi'_{e'}} \phi(w)$,
  and that the remaining block of~$\pi'_{e'}$ contains at least two
  elements. Thus, $e' = \{\phi(v), \phi(w)\}$, and since $e' \in \E'$,
  we have $\G \contained \G'$. 
\end{proof}

Given a graph~$\G$ with at least four vertices, none of which is
isolated, as well as an integer~$k \ge 4$, in order to decide
whether~$\G$ contains a $k$-clique, one simply constructs the
ASD's~$D(\set{K}_k)$ and $D(\G)$, and checks whether~$D(\set{K}_k) \le
D(\G)$. It is easy to see that the reduction is polynomial-time, and
this implies $\mathcal{NP}$-completeness\footnote{Of course, the
$k$-clique problem is still $\mathcal{NP}$-complete even when imposing
$k \ge 4$ and when looking at graphs with no isolated
vertices.}. Lemma~\ref{lem:red2} also implies that $D(\G) \equiv
D(\G')$ if and only if $\G \cong \G'$ for any two graphs~$\G, \G'$ as
in the statement of the lemma. Hence, deciding equivalence of ASD's is
at least as difficult as deciding graph isomorphism, since deciding
isomorphism is clearly not (computationally) easier when restricted to
such graphs. This completes the proof of
Theorem~\ref{thm:npcompleteness}.

We conclude this section by noting that one can provide a simple two-round
interactive proof for the problem of deciding non-equivalence of ASD's (see
Appendix~\ref{app:interact}).  This means that deciding non-equivalence is
in the complexity class~$\mathcal{IP}(2)$, and hence also
in~$\mathcal{AM}$~\cite{GS86}.  For this reason, if the problem of deciding
equivalence of ASD's were~$\mathcal{NP}$-complete, we would
have~$\mathcal{NP} \subseteq \textsf{co-}\mathcal{AM}$, and it is
well-known~\cite{BHZ87} that this implies a collapse of the polynomial
hierarchy~$\mathcal{PH}$ to its second level. Therefore, it is very
unlikely that deciding device equivalence is~$\mathcal{NP}$-complete.

\section{Binary ASD's and Unique Factorizations}
\label{sec:products}

We say that an ASD~$D$ has \emph{direct product factorization}
$\bigtimes_{i = 1}^m D_i$ if this product is equivalent to $D$.
Furthermore, an ASD $D$ is \emph{prime} if, whenever $D \equiv E \times
E'$, then either $E$ or $E'$ is trivial. For example, if $D$ is minimal
with a partition~$\pi \in \Partof{D}$ such that $|\pi| = p$ for a prime
number~$p$, then $D$ is prime.  Furthermore, every ASD~$D$ has a prime
factorization with at most $\log |\Stateof{D}|$ factors.

In the following, we look at the class~$\D^{\times}_{2}$ of ASD's having
(at least one) prime factorization consisting uniquely of binary ASD's.
Note that this class is closed under taking direct products. The following
lemma provides a strong necessary and sufficient condition for deciding
reducibility among members of the class~$\D^{\times}_{2}$ with the same
number of states, and such that no perfect factor appears in their binary
factorization. The reader is referred to
Appendix~\ref{app:binary_device_product} for a proof.

\begin{lemma}
  \label{lem:binary_device_product} Let $D_1, \ldots, D_m$, $D'_1,
  \ldots, D'_n$ be non-perfect state-minimal binary ASD's such that
  $\prod_{i = 1}^m |\Stateof{D_i}| = \prod_{j = 1}^n
  |\Stateof{D'_j}|$. Then $\bigtimes_{i = 1}^m D_i \le \bigtimes_{j =
  1}^n D'_j$ holds if and only if there exists a partition
  $\{J_1,\ldots, J_m\}$ of the indices $\{1,\ldots,n\}$ such that $D_i
  \le \bigtimes_{j \in J_i} D'_j$ for all $i \in \{1,\ldots,m\}$.
\end{lemma}

As a corollary of this fact, for given linear devices~$L_{k_1},
\ldots, L_{k_m}$, $L_{r_1}, \ldots, L_{r_n}$ with $\sum_{i = 1}^m
k_i = \sum_{j = 1}^n r_j$, we have~$\bigtimes_{i = 1}^m L_{k_i} \le
\bigtimes_{j = 1}^n L_{r_j}$ if and only if $m \le n$ and there exists
a partition~$\{J_1,\ldots, J_m\}$ of $\{1, \ldots, n\}$ such that $k_i
= \sum_{j \in J_i} r_j$. For instance, one can see that $L_3 \times
L_3 \nleq L_2 \times L_2 \times L_2$. Otherwise, the above would imply
that~$L_3 \le L_2$, which is obviously false.

The following theorem makes use of Lemma~\ref{lem:binary_device_product} to
show that the factorization in terms of \emph{binary} ASD's in unique.

\begin{theorem}
  \label{thm:product}
  Let $D$ be an ASD, and assume that~$\bigtimes_{i = 1}^m
  D_i$ is a factorization of~$D$ where $D_1, \ldots, D_m$ are
  \emph{binary}. Then, this factorization is unique (with respect to
  the set of all factorizations into binary devices), up to order and
  equivalence of the factors.
\end{theorem}
\begin{proof}
Let $D_1, \ldots, D_m, D'_1, \ldots, D'_m$ be binary ASD's such
that~$\bigtimes_{i = 1}^m D_i \equiv \bigtimes_{j = 1}^m D'_j$. In
order to prove the theorem, it suffices to show that these
factorizations are equivalent, that is, there exists a
permutation~$\gamma: \{1, \ldots, m\} \to \{1, \ldots, m\}$ such that
$D_i \equiv D'_{\gamma(i)}$ for all $i = 1,\ldots, m$. Without loss of
generality, assume that all devices are minimal.

First, note that for a minimal binary ASD's~$D$, we have $\bigvee
\Partof{D} = \id_{\Stateof{D}}$ whenever $D$ is perfect,
whereas~$\bigvee \Partof{D} = \{\Stateof{D}\}$ otherwise. Therefore,
if exactly~$\ell$ binary devices in the product~$\bigtimes_{i = 1}^m
D_i$ are perfect, then~$\Bigl| \bigvee \Partof{\bigtimes_{i = 1}^m
D_i} \Bigr| = \Bigl| \bigtimes_{i = 1}^m \bigl( \bigvee \Partof{D_i}\bigr) \Bigr| =
2^\ell$. For this reason, both products~$\bigtimes_{i = 1}^m D_i$ and
$\bigtimes_{j = 1}^m D'_j$ have exactly the same number of perfect
binary devices, otherwise they would not be equivalent by
Proposition~\ref{prop:minimal_equivalence}. Hence, we can rewrite both
products as
\begin{displaymath}
  C \times E_1 \times \cdots \times E_k \equiv C \times E'_1 \times
\cdots \times E'_k
\end{displaymath}
for some $k \le m$, non-perfect binary ASD's~$E_1, \ldots, E_k, E'_1,
\ldots, E'_k$, and a perfect ASD~$C$. By
Proposition~\ref{prop:minimal_equivalence}, there exist bijections
$\phi: \Stateof{C \times E_1 \times \cdots \times E_k} \to \Stateof{C
\times E'_1 \times \cdots \times E'_k}$ and $\alpha: \Partof{E_1
\times \cdots \times E_k} \to \Partof{E'_1 \times \cdots \times E'_k}$
such that
\begin{equation}
  \label{eq:equiv}
  \id_{\Stateof{C}} \times \pi = (\id_{\Stateof{C}} \times \alpha(\pi)
  )\circ \phi
\end{equation}
for all $\pi \in \Partof{E_1 \times \cdots \times E_k}$. This in
particular implies
\begin{multline*}
\id_{\Stateof{C}} \times \bigl\{ \Stateof{E_1 \times \cdots \times E_k} \bigr\}
= \bigvee \Partof{C \times E_1 \times \cdots \times E_k} \\ =
\Bigr(\bigvee \Partof{C \times E'_1 \times \cdots \times E'_k} \Bigl)
\circ \phi = \Bigl(\id_{\Stateof{C}} \times \bigl\{ \Stateof{E'_1 \times
  \cdots \times E'_k} \bigr\} \Bigr) \circ \phi.
\end{multline*}
For a fix~$s \in \Stateof{C}$ and any two~$e_0, e_1 \in \Stateof{E_1
\times \cdots \times E_k}$ we have $(s, e_0) \equiv_{\bigvee \Partof{C
\times E_1 \times \cdots \times E_k}} (s, e_1)$ by
Proposition~\ref{prop:minimal_equivalence}, and thus $\phi(s, e_0)
\equiv_{\bigvee \Partof{C \times E'_1 \times \cdots \times E'_k}}
\phi(s, e_1)$. In order for this to hold, there has to exist~$t \in
\Stateof{C}$ such that $\phi(s, e) = (t, e')$ for all $e \in
\Stateof{\bigtimes_{i = 1}^k E_i}$, where $e' \in
\Stateof{\bigtimes_{i = 1}^k E'_i}$.

Without loss of generality, we can assume that there exists a
bijection~$\tilde{\phi}: \Stateof{E_1 \times \cdots \times E_k} \to
\Stateof{E'_1 \times \cdots \times E'_k}$ such that~$\phi(s, e) = (s,
\tilde{\phi}(e))$, and therefore by~(\ref{eq:equiv}) we have
$\id_{\Stateof{C}} \times \pi = \id_{\Stateof{C}} \times (\alpha(\pi)
\circ \tilde{\phi})$ for all~$\pi$. This implies that $\pi =
\alpha(\pi) \circ \phi$, and thus $E_1 \times \cdots \times E_k \equiv
E'_1 \times \cdots \times E'_k$, again by
Proposition~\ref{prop:minimal_equivalence}. 

It now suffices to prove that these two last factorizations are
equivalent in order to conclude the proof. Note that since all devices
are non-perfect, both $\bigtimes_{i = 1}^k E_i \le \bigtimes_{i = 1}^k
E'_i$ and $\bigtimes_{i = 1}^k E'_i \le \bigtimes_{i = 1}^k E_i$
hold. By Lemma~\ref{lem:binary_device_product}, there exist
permutations $\gamma, \gamma'$ of $\{1,\ldots, k\}$ such that $E_i \le
E'_{\gamma(i)}$ and $E'_j \le E'_{\gamma'(j)}$.  Assume that there is
an $i \in \{1,\ldots,k\}$ such that $E_i$ and $E'_{\gamma(i)}$ are not
equivalent, i.e., $E'_{\gamma(i)} \nleq E_i$.  Define $\tilde{\gamma}
= \gamma' \circ \gamma$. Then, $E_i \le E'_{\gamma(i)} \le
E_{\tilde{\gamma}^r(i)}$ for all $r > 0$. Since $k$ is finite there
exists $r' > 0$ such that $\tilde{\gamma}^{r'}(i) = i$.  Therefore,
$E_i \equiv E'_{\gamma(i)}$, which is a contradiction.
\end{proof}

An immediate corollary of the theorem is the following.

\begin{corollary}
Two products of binary linear devices are equivalent if and only if
they consist of exactly the same devices.
\end{corollary}

For instance, the corollary immediately yields~$L_4 \times L_3 \times
L_3 \notequiv L_4 \times L_4 \times L_2$. Note that this
non-equivalence could not be proved using simpler arguments based on
order-preserving functions.

We stress that Theorem~\ref{thm:product} does not rule out the fact
that there might be additional factorizations in terms of non-binary
prime ASD's. Indeed, the general question of deciding whether prime
factorizations of ASD's are unique appears to be challenging. For
instance, it is easy to see that every perfect ASD~$C_m$ where $m =
\prod_{i = 1}^r p_i^{\alpha_i}$ for distinct primes $p_1, \ldots,
p_r$, and positive integers $\alpha_1, \ldots, \alpha_r$ can be
uniquely factorized as $\bigtimes_{i = 1}^r C_{p_i}^{\alpha_i}$. We
leave the more general question as an open problem. Note that the
problem is related to a line of research investigating unique
factorizations of {\em relational structures} (cf.\ e.g.~\cite{Jo66}
for a survey). Even though ASD's are related to relational structures,
known results only apply to a weaker form of direct product.

\section*{Acknowledgments}
This research was partially supported by the Swiss National Science
Foundation (SNF), project no. 200020-113700/1. We also thank Thomas
Holenstein for helpful discussions.

\appendix

\section{Direct Products of Set Partitions}
\label{app:partitions}
We prove here two facts about direct products of partitions. The first
proposition states that one can look at the refinement order component
wise.

\begin{proposition} \label{prop:productorder}
  Let~$\set{S}, \set{S}'$ be sets, $\pi, \rho \in
  \Part{\set{S}}$, and  $\pi', \rho' \in \Part{\set{S}'}$. Then
\begin{equation*}
  (\pi \times \pi') \refines (\rho \times \rho')
  \,\Longleftrightarrow\,(\pi \refines \rho) \land (\pi' \refines
  \rho').
\end{equation*}
\end{proposition}
\begin{proof}
  The proof follows from the fact that, given sets~$\set{B}, \set{B}',
  \set{C}$, and $\set{C}'$, we have $\set{B} \times \set{B}' \subseteq
  \set{C} \times \set{C}'$ if and only if $\set{B} \subseteq \set{C}$
  and $\set{B}' \subseteq \set{C}'$. If $\pi \refines \rho$ and $\pi'
  \refines \rho'$ both hold, then for every $\set{B} \in \pi$,
  $\set{B}' \in \pi'$ there have to exist $\set{C} \in \rho, \set{C}'
  \in \rho'$ such that $\set{B} \subseteq \set{C}$ and $\set{B}'
  \subseteq \set{C}'$, and hence $\set{B} \times \set{B}' \subseteq
  \set{C} \times \set{C}'$, which implies~$(\pi \refines \rho)\text{
  and }(\pi' \refines \rho')$. Conversely, if $(\pi \times \pi')
  \refines (\rho \times \rho')$, then for every $\set{B} \times
  \set{B}'$ there is $\set{C} \times \set{C}'$ such that $\set{B}
  \subseteq \set{B}'$ and $\set{C} \subseteq \set{C'}$. In particular,
  $\pi \refines \rho$ and $\pi' \refines \rho'$. 
\end{proof}

The second proposition states that the meet (join) of direct product
partitions is the direct product of the meets (joins).

\begin{proposition} 
  Let $\set{S}, \set{S}'$ be sets, $\pi, \rho \in
  \Part{\set{S}}$, and $\pi', \rho' \in
  \Part{\set{S}'}$. Then
  \begin{enumerate}[(i)]
  \item $(\pi \times \pi') \land (\rho \times \rho') = (\pi \land
    \rho) \times (\pi' \land \rho')$
  \item $(\pi \times \pi') \lor (\rho \times \rho') = (\pi \lor \rho)
    \times (\pi' \lor \rho')$
  \end{enumerate}
\end{proposition}
\begin{proof}
  For the first statement, we have directly
  \begin{equation*}
    \begin{split}
      (\pi \times \pi') \land (\rho \times \rho') & = \{(\set{B}
      \times \set{B}') \cap (\set{C} \times \set{C}')\,|\,\set{B} \in
      \pi, \set{B}' \in \pi', \set{C} \in \rho, \set{C}' \in \rho' \}
      \\ & = \{(\set{B} \cap \set{C}) \times (\set{B}' \cap
      \set{C}')\,|\,\set{B} \in \pi, \set{B}' \in \pi', \set{C} \in
      \rho, \set{C}' \in \rho' \}  \\ & = (\pi \land \rho) \times (\pi'
      \land \rho').
      \end{split}
  \end{equation*}
  To prove the second statement, first note that by definition $\pi
  \refines \pi \lor \rho$ and $\pi' \refines \pi' \lor \rho'$, and
  therefore $\pi \times \pi' \refines (\pi \lor \rho) \times (\pi'
  \lor \rho')$ by Proposition~\ref{prop:productorder}.  Analogously,
  $\rho \times \rho' \refines (\pi \lor \rho) \times (\pi' \lor
  \rho')$, which implies $(\pi \times \pi') \lor (\rho \times \rho')
  \refines (\pi \lor \rho) \times (\pi' \lor \rho')$.  

  Now, let $(s, s'), (t, t') \in \set{S} \times \set{S}'$ be such that
  $(s, s') \equiv_{(\pi \lor \rho) \times (\pi' \lor \rho')} (t,
  t')$. This implies that $s \equiv_{\pi \lor \rho} t$ and $s'
  \equiv_{\pi' \lor \rho'} t'$. There are $y_1, \ldots, y_k \in
  \set{S}$ with $s = y_1$ and $t = y_k$ such that for all $i =
  1,\ldots,k-1$ we have $y_i \equiv_{\pi} y_{i + 1}$ or $y_i
  \equiv_{\rho} y_{i + 1}$. Analogously, there are $y'_1, \ldots,
  y'_\ell \in \set{S}'$ with $s' = y'_1$ and $t' = y'_\ell$ such that
  for all $j = 1,\ldots,\ell-1$ we have $y'_j \equiv_{\pi'} y'_{j +
  1}$ or $y'_j \equiv_{\rho'} y'_{j + 1}$.  In particular, for all $i
  = 1, \ldots, k - 1$ we have $(y_i, s') \equiv_{\pi \times \pi'}
  (y_{i + 1}, s')$ or $(y_{i}, s') \equiv_{\rho \times \rho'} (y_{i +
  1}, s')$. Additionally, for all $j = 1, \ldots, \ell - 1$ we have
  $(t, y'_j) \equiv_{\pi \times \pi'} (t, y'_{j + 1})$ or $(t, y'_j)
  \equiv_{\rho \times \rho'} (t, y_{j + 1}')$. Therefore, $(s, s')
  \equiv_{(\pi \times \pi') \lor (\rho \times \rho')} (t, t')$. That
  is, $(\pi \lor \rho) \times (\pi' \lor \rho') \refines (\pi \times
  \pi') \lor (\rho \times \rho')$, and this implies equality. 
\end{proof}

\section{Interactive Proof for Device Non-Equivalence}
\label{app:interact}

In this section, we briefly sketch a two-round interactive proof for the
problem of non-equivalence of ASD's. The protocol follows the same lines as
the one for graph non-isomorphism.

Assume that Alice and Bob are given a pair of ASD's~$(D_0, D_1)$, and Alice
would like to prove~$D_0 \not\equiv D_1$ to Bob. Also, assume without loss
of generality that~$D_0$ and $D_1$ are minimal, and that~$\Stateof{D_0} =
\Stateof{D_1} = \{1, \ldots, n\}$ for some integer~$n$, and~$|\Partof{D_0}|
= |\Partof{D_1}|$.  Bob starts the protocol by choosing a bit~$b \in
\{0,1\}$ uniformly at random and generates an equivalent device~$D \equiv
D_b$ uniformly at random. (By Proposition~\ref{prop:minimal_equivalence},
this can be done efficiently by choosing an appropriate pair of
permutations~$(\phi, \alpha)$ uniformly at random.). He subsequently sends
the description of~$D$ to Alice.  Finally, Alice returns a bit~$b' \in
\{0,1\}$ to Bob, and Bob accepts if and only if~$b = b'$.

Whenever~$D_0 \not\equiv D_1$ holds, Alice is able to decide whether~$D_0
\equiv D$ or $D_1 \equiv D$, and hence to perfectly guess~$b$. However,
if~$D_0 \equiv D_1$, Alice can make Bob accept with probability at
most~$\frac{1}{2}$ regardless of her strategy.

\section{Proof of Lemma~\ref{lem:binary_device_product}}

\label{app:binary_device_product}
  Sufficiency is obvious. To prove the converse, assume that
  $\bigtimes_{i = 1}^m D_i \le \bigtimes_{j = 1}^n D'_j$, and let
  $(\phi, \alpha)$ be an arbitrary reduction of $\bigtimes_{i = 1}^m
  D_i $ to $\bigtimes_{j = 1}^n D'_j$. By Lemma~\ref{lem:minimality},
  the function $\phi$ is a bijection. In the following, we show that
  such a $\phi$ induces a partition of the set of indices
  $\{1,\ldots,n\}$ as in the statement of the theorem. To do this, we
  introduce the following function $\tau$: Let $j \in \{1,\ldots,n\}$
  and $(\state'_1, \ldots, \state'_{j-1}, \state_{j+1}', \ldots,
  \state'_n) \in \Stateof{D'_1} \times \cdots \times
  \Stateof{D'_{j-1}} \times \Stateof{D'_{j+1}} \times \cdots \times
  \Stateof{D'_n}$, then we define
  \begin{multline*}
    \tau(j, \state'_1, \ldots, \state'_{j-1}, \state'_{j+1}, \ldots, \state'_n) := \\
    \bigl\{j \,|\, |\phi^{-1}_j(\{\state'_1\} \times \cdots \times
    \{\state'_{j-1}\} \times \Stateof{D'_j} \times \{\state'_{j+1}\}
    \times \cdots \times \{\state'_n\} | > 1 \bigr\},
  \end{multline*}
  where $\phi^{-1}_j$ denotes the $j$-th component of the output of
  the function $\phi^{-1}$. In other words, the value~$\tau(j,
  \state'_1, \ldots, \state'_{j-1}, \state'_{j+1}, \ldots, \state'_n)$
  provides the set of indices of the devices in the
  product~$\bigtimes_{i = 1}^m D_i$ for which the state is modified
  when one goes over all possible states of the ASD~$D_j'$, fixing the
  states of the remaining ASD's to $\state'_1, \ldots, \state'_{j-1},
  \state'_{j + 1}, \ldots, \state'_n$, and looks at the output of
  $\phi^{-1}$.  We start by proving the following claim, which states
  that for a given~$j \in \{1,\ldots, n\}$, arbitrarily modifying the
  $j$-th component of a vector in~$\Stateof{D'_1} \times \cdots \times
  \Stateof{D'_n}$ only alters a single component of the output with
  respect to~$\phi^{-1}$.

  \begin{claim}
    \label{clm:clm1}
    For all~$j \in \{1, \ldots, n\}$ and $(\state'_1, \ldots,
    \state'_{j-1}, \state'_{j+1}, \ldots, \state'_n) \in \Stateof{D'_1}
    \times \cdots \times \Stateof{D'_{j - 1}} \times \Stateof{D'_{j+1}}
    \times \cdots \times \Stateof{D'_n}$, we have $|\tau(j,
    \state'_1,\ldots, \state'_{j-1}, \state'_{j+1}, \ldots, \state'_n)| =
    1$.
  \end{claim}
  \begin{proof}
    Assume, towards a contradiction, that the claim is false. In
    particular, there are states~$\tilde{\state}^1, \tilde{\state}^2
    \in \Stateof{D'_j}$ such that
    \begin{displaymath}
      \begin{split}
      (\state^1_1,\ldots, \state^1_m) & := \phi^{-1}(\state'_1,
    \ldots, \state'_{j-1}, \tilde{\state}^1, \state'_{j+1}, \ldots,
    \state'_n), \\ (\state^2_1, \ldots, \state^2_m) & :=
    \phi^{-1}(\state'_1, \ldots, \state'_{j-1}, \tilde{\state}^2,
    \state'_{j+1}, \ldots, \state'_n)
    \end{split}
    \end{displaymath}
    differ in two components $p$ and $q$, that is, $\state^1_p \ne
    \state^2_p$ and $\state^1_q \ne \state^2_q$. Since
    $|\Stateof{D_j'}| \ge 3$ by our assumption, pick an arbitrary
    third element $\tilde{\state}^3 \in \Stateof{D'_j}$ different from
    $\tilde{\state}^1$ and $\tilde{\state}^2$, and define
    $(\state^3_1, \ldots, \state^3_m) := \phi^{-1}(\state'_1, \ldots,
    \state'_{j-1}, \tilde{\state}^3, \state'_{j+1}, \ldots,
    \state'_n)$.  We look for partitions $\pi_1, \ldots, \pi_m$, where
    $\pi_i \in \Partof{D_i}$, such that the vectors $(\state^1_1,
    \ldots, \state^1_n)$, $(\state^2_1, \ldots, \state^2_n)$, and
    $(\state^3_1, \ldots, \state^3_n )$ are each in a distinct block
    of $\pi_1 \times \cdots \times \pi_m$.  In order to do so,
    consider the following two cases:

    \begin{enumerate}[(i)]
    \item $\state^1_p = \state^3_p$: Choose a partition $\pi_p \in
      \Partof{D_p}$ such that $\state^1_p \not\equiv_{\pi_p}
      \state^2_p$ (this exists by state-minimality). Furthermore,
      there must exist a component $r \ne p$ such that $\state^1_r \ne
      \state^3_r$. Then, simply pick $\pi_r \in \Partof{D_r}$ such
      that $\state^1_r \not\equiv_{\pi_r} \state^3_r$. All $\pi_i$
      for $i \ne p$ and $i \ne r$ can be chosen arbitrarily.  The
      cases $\state^1_q = \state^3_q$, $\state^2_p = \state^3_p$, and
      $\state^2_q = \state^3_q$ are analogous. (Notice that these
      cases are not mutually-exclusive.)
    \item $\state^3_p \ne \state^1_p$, $\state^3_p \ne \state^2_p$,
      $\state^3_q \ne \state^1_q$, and $\state^3_q \ne \state^2_q$:
      Choose a partition $\pi_p \in \Partof{D_p}$ such that
      $\state^1_p \not\equiv_{\pi_p} \state^3_p$. Now, it might be
      that either $\state^2_p \equiv_{\pi_p} \state^3_p$ or
      $\state^2_p \equiv_{\pi_p} \state^1_p$. In the first case,
      choose $\pi_q \in \Partof{D_q}$ such that $\state^2_q
      \not\equiv_{\pi_q} \state^3_q$, whereas in the second case
      choose $\pi_q \in \Partof{D_q}$ such that $\state^2_q
      \not\equiv_{\pi_q} \state^1_q$. All $\pi_i$ for $i \ne p$ and $i
      \ne q$ are chosen arbitrarily.
    \end{enumerate}
    Since $D'_j$ is binary, there are distinct $u,v \in \{1,2,3\}$
    such that
    \begin{displaymath}
      (\state'_1, \ldots, \state'_{j-1}, \tilde{\state}^u,
      \state'_{j+1}, \ldots, \state'_n) \equiv_{\alpha(\pi_1 \times
      \cdots \pi_m)} (\state'_1, \ldots, \state'_{j-1},
      \tilde{\state}^v, \state'_{j+1}, \ldots, \state'_n).
    \end{displaymath}
    However, we have $(\state^u_1, \ldots, \state^u_m)
    \not\equiv_{\pi_1 \times \cdots \times \pi_m} (\state^v_1, \ldots,
    \state^v_m)$, and $(\phi, \alpha)$ cannot be a reduction. 
  \end{proof}

  We now want to prove that the unique component which varies is
  independent of the other states.
  \begin{claim}
    \label{clm:clm2}
    For all~$j \in \{1, \ldots, n\}$ and for all
    $(\state_{1,1}', \ldots, \state_{1,j-1}', \state_{1,j+1}', \ldots,
    \state'_{1,n})$,\\ $ (\state_{2,1}', \ldots, \state_{2,j-1}',
    \state_{2,j+1}', \ldots, \state'_{2,n}) \in \Stateof{D'_1} \times
    \cdots \times \Stateof{D'_{j - 1}} \times \Stateof{D'_{j+1}}
    \times \cdots \times \Stateof{D'_n}$, we have
    \begin{equation*}
      \tau(j, \state'_{1,1}, \ldots, \state'_{1,j-1}, \state'_{1,j+1},
      \ldots, \state'_{1,n}) = \tau(j, \state'_{2,1},
      \ldots, \state'_{2,j-1}, \state'_{2,j+1}, \ldots, \state'_{2,n}).
    \end{equation*}
  \end{claim}
  \begin{proof}
    We start by assuming that there is a unique component $r \in
    \{1,\ldots,n\} - \{j\}$ such that $\state_{1,r}' \ne
    \state_{2,r}'$. Also, assume without loss of generality that $j
    \ne n$ and $r = n$. In particular, denote $\state'_{j'} :=
    \state_{1,j'} = \state_{2,j'}$ for all $j' \in \{1,\ldots,n-1\} -
    \{j\}$. Given two states $\state_{A,j}', \state_{B,j}' \in
    \Stateof{D'_j}$, we define the following states of $D'_1 \times
    \cdots \times D'_n$:
    \begin{displaymath}
      \begin{split}
        \state_{1,A}' & := (\state'_{1}, \ldots, \state'_{j-1},
        \state_{A,j}', \state'_{j+1},
        \ldots, \state'_{n-1}, \state'_{1,n}), \\
        \state_{1,B}' & := (\state'_{1}, \ldots, \state'_{j-1},
        \state_{B,j}', \state'_{j+1},
        \ldots, \state'_{n -1}, \state'_{1,n}), \\
        \state_{2,A}' & := (\state'_{1}, \ldots, \state'_{j-1},
        \state_{A,j}', \state'_{j+1},
        \ldots, \state'_{n - 1}, \state'_{2,n}), \\
        \state_{2,B}' & := (\state'_{1}, \ldots, \state'_{j-1},
        \state_{B,j}', \state'_{j+1}, \ldots, \state'_{n-1},
        \state'_{2,n}).
      \end{split}
    \end{displaymath}
    Furthermore, using the previous claim, we define
    \begin{displaymath}
      \begin{split}
        \{p_1 \}  & := \tau(j, \state'_{1}, \ldots,
        \state'_{j-1}, \state'_{j+1}, \ldots, \state'_{n-1},
        \state'_{1,n}) \\
        \{p_2 \} & := \tau(j, \state'_{1}, \ldots, \state'_{j-1},
        \state'_{j+1}, \ldots, \state'_{n-1}, \state'_{2,n}),
      \end{split}
    \end{displaymath}
    as well as
    \begin{displaymath}
      \begin{split}
        \{r_A\} & := \tau(n, \state'_1, \ldots, \state'_{j-1},
        \state_{A,j}, \state'_{j+1}, \ldots, \state'_{n-1}) \\
        \{r_B\} & := \tau(n, \state'_1, \ldots, \state'_{j-1},
        \state_{B,j}, \state'_{j+1}, \ldots, \state'_{n-1}).
      \end{split}
    \end{displaymath}
    Finally, we define the following states of $D_1 \times \cdots
    \times D_m$ making use of $\phi^{-1}$:
    \begin{displaymath}
      w := \phi^{-1}(\state_{1,A}'), x := \phi^{-1}(\state_{1,B}'), y := \phi^{-1}(\state_{2,B}'), z := \phi^{-1}(\state_{2,A}').
    \end{displaymath}
    We have $w_{p_1} \ne x_{p_1}$ and $x_p = w_p$ for all $p \in
    \{1,\ldots, m\} - \{p_1\}$. Analogously $y_{p_2} \ne z_{p_2}$ and
    $y_p = z_p$ for all $p \in \{1,\ldots, m\} - \{p_2\}$.  And again,
    by the same argument, $x$ and $y$ differs only at component $r_A$,
    and $z$ and $w$ differ only in component $r_B$.  According to
    this, there are two ways to modify state $w$ into $x$. The first
    one is by changing component $p_1$. The second one is by going
    through states $z$ and $y$, modifying components $r_A, p_2,$ and
    $r_B$. Assume, towards a contradiction, that $p_1 \ne p_2$. Since
    $w_{p_2} = x_{p_2}$, we must have either $r_A = p_2$ and $r_B =
    p_1$ or $r_A = p_1$ and $r_B = p_2$. If the former holds, we
    necessarily have $w = y$, while if the latter holds, then $z =
    x$. In both cases, we have a contradiction with the fact that
    $\phi$ is a bijection.
    
    The proof of the claim easily follows by repeating the same argument
    iteratively for $r \ne n$. 
  \end{proof}

  Hence, for $j \in \{1, \ldots, n\}$ we are now allowed to denote by
  $\tau(j)$ the unique component which varies when altering the state
  of $D'_j$.  Additionally, for $i \in \{1,\ldots, m\}$, we define
  \begin{displaymath}
    \rho(i) := \left\{ j \,|\, \tau(j) = i \right\}.
  \end{displaymath}
  Note that $\{\rho(i)\,|\, i =1,\ldots,m\} = \ker{\tau}$ is a set
  partition of $\{1,\ldots, n\}$. Now, take a fix $i \in \{1,\ldots,m\}$
  and let $\rho(i) = \{j_1, \ldots, j_r\}$.  Furthermore, fix states
  $\state_1, \ldots, \state_{i-1}, \state_{i+1},\ldots,\state_m$ for the
  devices $D_1, \ldots, D_{i-1}, D_{i+1},\ldots,D_m$. The $j$-th component
  $\phi_j(\state_1,\ldots,
  \state_{i-1},\state,\state_{i+1},\ldots,\state_n)$ is constant for all $j
  \notin \rho(i)$, since for any two states $\state', \state'' \in
  \Stateof{D'_1 \times \cdots \times D'_n}$ differing at two distinct
  components $p \ne q$ such that $\tau(p) \ne \tau(q)$, the states
  $\phi^{-1}(\state), \phi^{-1}(\state')$ differ at both components
  $\tau(p)$ and $\tau(q)$.

  Furthermore, fix arbitrary partitions~$\pi_1, \ldots, \pi_{i-1},
  \pi_{i+1}, \ldots, \pi_{n}$ for the devices $D_1$, $\ldots$, $D_{i-1}$,
  $D_{i+1}$, $\ldots$, $D_n$, and finally define
  {\small \begin{displaymath}
    \begin{split}
      \phi^i(\state) & := \left(\phi_{j_1}(\state_1,\ldots,
      \state_{i-1},\state,\state_{i+1},\ldots,\state_n), \ldots,
      \phi_{j_r}(\state_1,\ldots,
      \state_{i-1},\state,\state_{i+1},\ldots,\state_n) \right) \\
     \alpha^i(\pi) & := (\alpha_{j_1}(\pi_{1},\ldots,
      \pi_{i-1},\pi,\pi_{i+1},\ldots,\pi_n), \ldots, \alpha_{j_r}(\pi_1,\ldots,
      \pi_{i-1},\pi,\pi_{i+1},\ldots,\pi_n)).
    \end{split}
  \end{displaymath}}%
  It is now easy to verify that $(\phi^i, \alpha^i)$ is a reduction of
  $D_i$ to $\bigtimes_{j \in \rho(i)} D'_j$. If this is not the case,
  there are distinct states $\state_i, \state_i' \in D_i$ and $\pi_i
  \in \Partof{D}$ such that $\state_i \not\equiv_{\pi_i} \state_i'$,
  but $\state_i =_{\alpha^i(\pi_i) \circ \phi^i} \state'_i$. By the
  arguments above, and by the definition of $\phi^i$ and $\alpha^i$,
  this implies that $(\phi, \alpha)$ is not a reduction.
\end{document}